\documentclass[conference]{IEEEtran}
\IEEEoverridecommandlockouts
\usepackage{cite}
\usepackage{epstopdf}
\usepackage{amsmath,amssymb,amsthm,mathrsfs,amsfonts,dsfont}
\usepackage{algorithmic}
\usepackage{graphicx}
\usepackage{graphics}
\usepackage{color}
\usepackage{textcomp}
\usepackage{balance}
\usepackage{tikz}
\usepackage[top=0.75in,bottom=1in,left=0.675in,right=0.675in]{geometry}
\newcommand{\site}{\boldsymbol{s}}

\newtheorem{prop}{Proposition}
\def\BibTeX{{\rm B\kern-.05em{\sc i\kern-.025em b}\kern-.08em
		T\kern-.1667em\lower.7ex\hbox{E}\kern-.125emX}}
\begin{document}
	
	\title{How To Dimension Radio Resources When Users Are Distributed on Roads Modeled by Poisson Line Process}

	\author{\IEEEauthorblockN{Jalal Rachad$~^{*+}$, Ridha Nasri$~^*$, Laurent Decreusefond$~^+$}
		\IEEEauthorblockA{$~^*$Orange Labs: Direction of Green transformation, Data knowledge, traffic and resources Modelling,\\ 40-48 avenue de la Republique
			92320 Chatillon, France\\ 
			$~^+$LTCI, Telecom ParisTech, Universite Paris-Saclay,\\
			 75013, Paris, France\\
			Email:\{$^{*}$jalal.rachad,  $~^*$ridha.nasri\}@orange.com, $^+$Laurent.Decreusefond@mines-telecom.fr}}

%
%
%
	\maketitle
	\begin{abstract}
		Resources dimensioning aims at finding the number of radio resources required to carry a forecast data traffic at a target users Quality of Services (QoS). The present paper attempts to provide a new approach of radio resources dimensioning considering the congestion probability, qualified as a relevant metric for QoS evaluation. Users are assumed to be distributed according to a linear Poisson Point Process (PPP) in a random system of roads modeled by Poisson Line Process (PLP) instead of the widely-used spatial PPP. We derive the analytical expression of the congestion probability for analyzing its behavior as a function of network parameters. Finally we show how to dimension radio resources by setting a value of the congestion probability, often targeted by the operator, in order to find the relation between the necessary resources and the forecast data traffic expressed in terms of cell throughput. Different numerical results are presented to justify this dimensioning approach.
		 
	\end{abstract}
	
	
	\section{Introduction}

    Radio resources dimensioning is a major task in cellular network development that aims at maximizing the network resources efficiency. Dimensioning process consists in assessing the necessary resources that permit to carry a predicted data traffic in order to satisfy a given QoS. Therefore, an accurate traffic estimation is one of the major issues to consider during dimensioning process. Also, it is of utmost importance to choose a relevant QoS metric to evaluate system performances.
    
    On the other hand, 5G-NR (5th Generation-New Radio) interface inherits many concepts and techniques from 4G systems. It comes with the scalable OFDM (Orthogonal Frequency Division Multiplex) technology having different subcarriers' spacing ($\Delta f=2^{\nu}15$ kHz, where $\nu=0$ to $4$) and diverse spectrum bands \cite{3GPPTS38214}. A set of OFDM subcarriers constitutes the basic unit of radio resources that a cell can allocate to a mobile user. This set of subcarriers is called Physical Resource Block (PRB) in the 3rd Generation Partnership Project (3GPP) terminology. Moreover, the allocation of PRBs to users is performed at each Time Transmit Interval (TTI) according to a predefined scheduling algorithm. The choose of this latter is mainly related to the fairness level between users, i.e., the way that resources are allocated to users according to their channel qualities and their priorities, defined by the operator \cite{trivedi2014comparison}.\\
    

    
    The available scientific literature of OFDMA based systems' dimensioning is quite rich and many aspects have been discussed; see for instance \cite{agarwal2008low, khattab2006opportunistic, decreusefond2012robust, shen2005adaptive}. It was provided in \cite{decreusefond2012robust} an analytical model for dimensioning OFDMA based networks with proportional fairness in resource allocation between users requiring different transmission rates. For a Poisson distribution of mobile users, \cite{decreusefond2012robust} showed that the required number of resources in a typical cell follows a compound Poisson distribution. In addition, an upper bound of the outage probability was given. In \cite{shen2005adaptive}, an adaptive resource allocation for multi-user OFDM system, with a set of proportional fairness constraints guaranteeing the required data rate, has been discussed.
    
    Furthermore, the network geometry and the spatial distribution of users are important factors to consider when it comes to system performance analysis and dimensioning problem. Different models for network geometry and user distributions can be found in \cite{nasri2016analytical, andrews2011tractable, chetlur2018coverage, choi2018analytical}. In particular, authors in \cite{chetlur2018coverage} and \cite{choi2018analytical} considered vehicular-type communication systems where the transmitting and receiving nodes are distributed along roads, modeled by a PLP. It seems that PLP is a relevant model for roads in urban environment that is gaining popularity recently  and merits investigations when looking for performance analysis and dimensioning process of wireless cellular communications.\\


  Compared to the existing works, the main contribution of this paper is to provide  an analytical model to dimension OFDM based systems with a proportional fair resources' allocation policy. This dimensioning model is very useful for operators because it gives a vision on how they should manage the available spectrum. If the dimensioned number of resources exceeds the available one, the operator can, for instance,  aggregate fragmented spectrum resources into a single wider band in order to increase the available PRBs, or activate capacity improvement features like dual connectivity between 5G and legacy 4G networks, in order to delay investment on the acquisition of new spectrum bands. Moreover, the proposed model can be applied to the scalable OFDM based 5G NR with different subcarriers' spacing in order to enable different types of deployments and network topologies and support different use cases. Additionally, We show that the total number of the requested PRBs follows a compound Poisson distribution and we derive the explicit formula of the congestion probability, that remains valid with every stochastic process describing the random tessellations of roads, as a function of different system parameters. This metric is defined as the risk that the requested resources exceed the available ones. It is often considered primordial for operators when it comes to resources dimensioning since it is related to the guaranteed quality of service.\\

    The rest of this paper is organized as follows: In Section II, system models, including a short description of Poisson Line Process, are provided. Section III characterizes the proposed dimensioning model and provides an explicit expression of the congestion probability and an implicit relation between the number of required resources and the cell throughput. Numerical results are provided in Section IV. Section V concludes the paper. 
    
    \section{System model and notations}
    
  Cellular networks modeling is often related to the network geometry, the shape of the cell, the association between cells and users and of course their spatial distribution. This latter is related to the geometry of the city where the studied cell area exists. The geometry of the city, in turn, is linked to the spatial distribution of roads and buildings. Indoor users, which are distributed in buildings, are often modeled by a spatial PPP in $\mathbb{R}^2$. However, outdoor users are always distributed along roads. To model the spatial distribution of roads, many models have been proposed in literature such as Manhattan model that uses a grid of horizontal and vertical streets \cite{bai2003important}, Poisson Voronoi Tessellations (PVT) \cite{baccelli2001coverage} and PLP that characterize the spatial random tessellations of roads. PVT does not fit the geometry of urban environments and could not lead to analytical results. So we choose to model, in this work, the random tessellations of roads by the PLP which is somehow realistic enough and surely tractable.\\

     
    \subsection{Poisson Line Process} 
    A PPP in $\mathbb{R}^2$ with intensity $\zeta$ is a point process that satisfies: \textit{i}) the number of points inside every bounded closed set $B \in \mathbb{R}^2$ follows a Poisson distribution with mean $\zeta|B|$, where $|B|$ is the Lebesgue measure on $\mathbb{R}^2$; \textit{ii}) the number of points inside any disjoint sets of $\mathbb{R}^2$ are independent \cite{nasri2015tractable}.
    
    The PLP is mathematically derived from the spatial PPP. Instead of points, the PLP is a random process of lines distributed in the plane $\mathbb{R}^2$. Each line in $\mathbb{R}^2$ is parametrized in terms of polar coordinates ($r$,$\theta$) obtained from the orthogonal projection of the origin on that line, with $r$ $\in$ $\mathbb{R^+}$ and $\theta \in (-\pi, \pi]$. Now we can consider an application $T$ that maps each line to a unique couple ($r$,$\theta$), generated by a PPP in the half-cylinder $\mathbb{R^+} \times (-\pi,\pi]$. The distribution of lines in $\mathbb{R}^2$ is the same as points' distribution in this half-cylinder; see \cite{chetlur2018coverage} and \cite{choi2018analytical} for more details.\\ 
    
     In the sequel, we assume that roads are modeled by a PLP $\phi$ with roads' intensity denoted by $\lambda$. The number of roads that lie inside a disk $\site$ of radius $R$ is a Poisson random variable, denoted by $Y$. It corresponds to the number of points of the equivalent PPP in the half-cylinder $[0,R] \times (-\pi,\pi]$ having an area of $2\pi R$. Hence, the expected number of roads that lie inside $\site$ is $\mathbb{E}(Y)=2 \pi \lambda R$. Moreover, users are assumed to be distributed on each road according to independents linear PPP with the same intensity $\delta$. This model is known as Cox point process. The mean number of users on a given road $j$ is $\delta L_j$, with $L_j$ is the length of road $j$. Besides, the number of roads that lie between two disks of radius $R_1$ and $R_2$ respectively, with $R_1 \leqslant R_2$, is $2\pi \lambda (R_2-R_1)$. Also, the number of distributed users in a road, parametrized by ($r$,$\theta$) and delimited by the two disks, is $2\delta (\sqrt{R_2^2 - r^2}-\sqrt{R_1^2 - r^2})$.
     
     Furthermore, the average number of users in the disk of radius $R$ is calculated using the equivalent homogeneous PPP with intensity $\lambda \delta$ in the disk area. Let $u$ denotes the average number of users inside the disk $\site$, it is written by 
     
     \begin{equation}
     u= \lambda \delta \pi R^2.
     \label{nbrusers}
     \end{equation}
     
     For illustration, Fig. \ref{plp} shows a realization of the described Process.  
        
        \begin{figure}[tb]
        	\centering
        	
        	\includegraphics[width=0.40\textwidth]{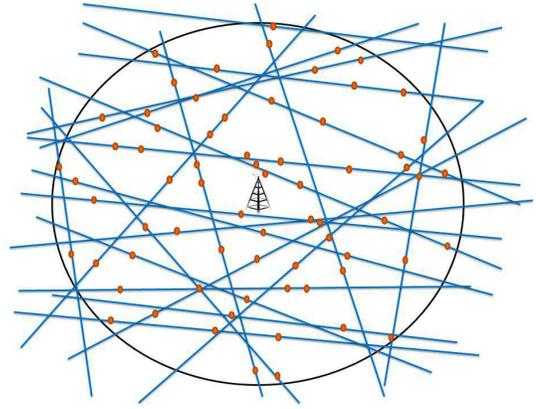}
        	\caption{A realization of Cox Point Process driven by PLP.}
        	\label{plp}
        \end{figure}  
     
%
    
    \subsection{Network model}
    We consider a circular cell $\site$ of radius $R$ with a Base Station (BS), denoted also $\site$ and positioned at its center, transmitting with a power level $P$. The received power by a user located at distance $x$ from $\site$ is $Px^{-2b}/a$, where $2b$ is the path loss exponent and $a$ is the propagation constant. We assume that BS $\site$ allocates PRBs to its users at every TTI (e.g., 1 ms). Each PRB has a bandwidth denoted by $W$ (e.g., $W=$180kHz for scalable OFDM with subcarriers spacing of 15kHz). Active users in the cell compete to have access to the available dimensioned PRBs. Their number is denoted by $M$. The BS allocates a given number $n$ of PRBs to a given user depending on: \textit{i}) the class of services he belongs to (i.e.,the transmission rate he requires) and \textit{ii}) his position in the cell (i.e., the perceived radio conditions). Without loss of generality, we assume that there is just one class of service with a required transmission rate denoted by $C^*$.
    
    A user located at distance $x$ from $\site$ decodes the signal only if the metric \textquotedblleft Signal to Interference plus Noise Ratio (SINR)\textquotedblright $~\Theta(x)= \frac{P x^{-2b}/a}{I+\sigma^2}$ is above a threshold $\Theta^*=\Theta(R)$, where $I$ is the received co-channel interferences and $\sigma^2$ is the thermal noise power. For performance analysis purpose, the SINR $\Theta(x)$ is often mapped to the user throughput by a link level curve. For simplicity of calculation, we use hereafter the upper bound of the well known Shannon's formula for MIMO system $Tx \times Rx $, with $Tx$ and $Rx$ are respectively the number of transmit and receive antennas. Hence, the throughput of a user located at distance $x$ from $\site$ is   
    
     \begin{equation}
    C(x)=\vartheta W\log_2\left(1+\Theta(x)\right), 
    \label{Shannonmimo}
    \end{equation}
    with $\vartheta= \min(Tx,Rx)$.\\
    
     Then, the number of PRBs required by a user located at distance $x$ from $\site$ is
   \begin{equation}
   n(x)=\lceil \frac{C^*}{C(x)} \rceil \leq N,
   \label{nbrprb}
   \end{equation}
   where
   
  \begin{equation}
  N= \min(N_{max} , \lceil C^*/(\vartheta Wlog_2(1+\Theta^*))\rceil),
  \end{equation} 
  with $N_{max}$ is the maximum number of PRBs that a BS can allocate to a user (fixed by the operator) and $\lceil . \rceil$ stands for the Ceiling function.
  
  It is obvious from (\ref{nbrprb}) that users are fairly scheduled because a user with bad radio condition (i.e., with low value of $C(x)$) gets higher number of PRBs to achieve its transmission rate $C^*$.
  
   Let $d_n$ be the distance from $\site$ that verifies, for all $x \in (d_{n-1},d_n]$, $n(x)=n$ with
    
    \begin{equation}
    n= \frac{C^*}{C(d_n)}  
    \label{nbrprb2}
    \end{equation}
   is an integer and\\ 
   \[
   d_n=\left\{
   \begin{array}{ll}
    0 \ \mbox{if $n=0$,}  \\

      \\
   \left[\frac{a(I+\sigma^2)}{P}(2^{\frac{C^*}{n \vartheta W}}-1)\right]^{\frac{-1}{2 b}}\ \mbox{otherwise,}\\
   \\
   \end{array}
   \right.
   \]

   From (\ref{nbrprb2}), the cell area $\site$ can be divided into rings with radius $d_n$ such that for $1\leqslant n\leqslant N,~  0\leqslant d_{n-1}<d_n\leqslant R$. The area between the ring of radius $d_n$ and the ring of radius $d_{n-1}$ characterizes the region of the cell where users require $n$ PRBs to achieve the transmission rate $C^*$. We define the cell throughput by the sum over all transmission rates of users:    
    \begin{equation}
    \tau= u C^*,
    \label{cellth}
    \end{equation}
    with $u$ is recalled the average number of users inside $\site$ and expressed by (\ref{nbrusers}).\\

 On the other hand, inter-cell interference is one of the main factors that compromise cellular network performances. The analysis of this factor level go through the SINR evaluation that depends on the geometry of the network as well as the distribution of user locations. To model interference impact, we use in this paper the notion of interference margin (IM) or Noise Rise which is always used in link budget. IM is defined as the increase in the thermal noise level caused by other-cell interference. IM can be expressed in the linear scale as
 
 \begin{equation}
 IM=\frac{I+\sigma^2}{\sigma^2}
 \label{NR}
 \end{equation}

   \section{Presentation of the dimensioning approach}
   Dimensioning process consists in evaluating the required radio resources that allow to carry a forecast data traffic given a target QoS. The QoS can be measured by the congestion probability metric or even by a target average user throughput. The present approach assesses the congestion probability as a function of many key parameters, in particular the number of PRBs $M$ and the cell throughput $\tau$. To characterize this congestion probability, we need to evaluate the total requested PRBs by all users. In the remainder of this section, we will state some analytical results regarding the explicit expression of the congestion probability under the system model presented in the previous section.
     
   \subsection{Qualification of the number of total requested PRBs}
    
    As we have mentioned in section II.A,  users are distributed on each road $L_j$ according to a linear PPP of intensity $\delta$. Now, if we consider a disk B(0,$d_n$) of radius $d_n$, the number of users in the portion of $L_j$ that lies inside B(0,$d_n$) is a Poisson random variable (this comes from the definition of linear PPP) with mean $2 \delta \sqrt{d_n^2-r_j^2}$ (Pythagoras' theorem). Hence, conditionally on the PLP, the mean number of users $\alpha_n$ inside B(0,$d_n$) is the sum over all roads $L_j$ that intersect with B(0,$d_n$). It can be expressed by    
    
    \begin{equation}
    \alpha_{n}=2 \delta \sum_{j=1}^{Y}\mathds{1}_{(d_n>r_j)}\sqrt{d_n^2-r_j^2}.
    \label{alphan}
    \end{equation}
    
    Moreover, the number of users in the portion of $L_j$ that lies between two rings B(0,$d_n$) and B(0,$d_{n-1}$) is also a Poisson random variable with parameter (i.e, the mean number of users)  $2\delta (\sqrt{d_n^2 - r_j^2}-\sqrt{d_{n-1}^2 - r_j^2})$. Finally, the mean number of users $\mu_n$ in all roads that lie between the rings B(0,$d_n$) and B(0,$d_{n-1}$) can be expressed by 
    \begin{equation}
    \mu_n=\alpha_n - \alpha_{n-1}.
    \label{mun}
    \end{equation}
    
    To qualify the number of requested PRBs by users, we consider a Poisson random variable denoted by $X_n$  with parameter $\mu_n$. $X_n$ represents the number of users that request n PRBs with $1\leq n\leq N$. Hence, we can define the total number of requested PRBs in the cell as the sum of demanded PRBs in each ring. It can be expressed by
    
    \begin{equation}
    \Gamma=  \sum_{n=1}^{N} nX_n. 
    \label{prbtot}
    \end{equation}
    
    The random variable $\Gamma$ is the sum of weighted Poisson variables and its distribution is known as compound Poisson distribution. The evaluation of this distribution requires extensive numerical simulation. However, we fortunately derive here an analytical formula.

%
%
    
    
    \subsection{Congestion probability and dimensioning}
    The congestion probability, denoted by $\Pi$, is defined as the probability that the number of total requested PRBs in the cell is greater than the available PRBs fixed by the operator. In other words, it measures the probability of failing to achieve an output number of PRBs $M$ required to guarantee a predefined quality of services:
    
    \begin{equation}
   \Pi(M,\tau)= \mathbb{P}(\Gamma \geq M).
    \label{poutdef}
    \end{equation}
    
    The following proposition gives the explicit expression of the congestion probability conditionally on the PLP.
    \begin{prop} 
    Let $\Gamma$ be defined as in (\ref{prbtot}) and $\phi$ be a PLP defined as in section II.A, the probability that $\Gamma$ exceeds a threshold $M$ conditionally on $\phi$ is
      {
      	\begin{align}
      	\mathbb{P}(\Gamma\geq M|\phi)&= 1- \frac{1}{\pi} e^{-\delta ~\alpha_{N}} \times\nonumber\\
      	& \int_{0}^{\pi} e^{p_n(\theta)} \frac{\sin(\frac{M\theta}{2})}{\sin(\frac{\theta}{2})}\cos(\frac{M-1}{2}-q_n(\theta)) d\theta,
      	\label{theee}
      	\end{align}
      }
    where
     \[
    p_n(\theta)=\sum_{n=1}^{N}\mu_n \cos(n\theta)\text{ and } \\
    q_n(\theta)=\sum_{n=1}^{N}\mu_n \sin(n\theta).\\
    \]
    \label{theorem1}
    \end{prop}
    
    \begin{proof}
    See appendix \ref{proof1}.	
    \end{proof}

    Proposition \ref{theorem1} is valid not only for the considered Cox process but also for every process of user distribution, in particular spatial PPP model for which the parameters $\mu_n$ in (\ref{mun}) are adapted to $\mu_n= u (d_n^2 - d_{n-1}^2 )/R^2$. Also, this formula can be applied to every stochastic process that describes the random tessellation of roads (e.g., PVT...) once we have the  number of roads that lie inside the cell coverage area, their lengths and their distances from the serving BS.\\
    
    The final expression of the congestion probability can be derived by averaging over the PLP as follows:
    \begin{equation}
    \Pi(M,\tau) = \mathbb{E}_\phi[\mathbb{P}(\Gamma \geq M|\phi)].
    \end{equation}
    
    Once we have the expression of the congestion probability, we set a target congestion probability $\Pi^*$. The required number of PRBs $M$ is written as a function of $\tau$ through the implicit equation $\Pi(M,\tau)=\Pi^*$. The output $M$ of the implicit function constitutes the result of the dimensioning problem.
     
    \section{Numerical results}
    For numerical purpose, we consider a cell of radius $R=0.6km$ with power level $P=60dBm$ and operating in a bandwidth of $20MHz$. The downlink thermal noise power including the receiver noise figure is calculated for $20MHz$ such that $\sigma^2=-93dBm$. Since only users on roads are concerned, outdoor environment with propagation parameter $a=130dB$ and path loss exponent $2b=3.5$ are considered. We assume also that we have 8TX antennas in the BS and 2RX antennas in users' terminals. So, the number of possible transmission layers is at most 2. The SINR threshold is set to $\Theta^*=-10dB$.\\ 
    

    In Fig. \ref{ths}, we simulate the described Cox process driven by PLP and the spatial PPP in MATLAB for two values of cell throughput $\tau=8Mbps$ and $\tau=25Mbps$. We notice that the explicit expression of the congestion probability fits the empirical one obtained by using Monte-Carlo simulations. Moreover, it is obvious that an increase in cell throughput $\tau$ generates an increase in the congestion probability because $\tau$ is related to the number of users in the cell. When the intensity of users in roads $\delta$ increases or the intensity of roads $\lambda$ increases, the number of required PRBs on roads that lie inside the cell coverage area increases, thus the system experiences a high congestion. An other important factor that impacts system performances is the path loss exponent. The variations of this parameter has tremendous effect on the congestion probability: when $2b$ goes up, radio conditions become worse and consequently the number of demanded PRBs, to guarantee the required QoS, increases.\\

    Also, Fig. \ref{ths} shows the congestion probability obtained with the described Cox model and its comparison with the one of spatial PPP model for two values of the average cell throughput $\tau$. We observe that the number of requested PRBs by users is always higher, for every target value of the congestion probability, when the random tessellation of roads is taken in consideration. In other words, even if the mean number of users in the cell is the same, the geometry of the area covered by this cell has a significant impact on system performances. Furthermore, if we consider a Cox model with high roads intensity, users appear to be distributed every where in the cell as in spatial PPP model. In this case, the Cox model driven by the PLP can be approximated by a spatial PPP with the same intensity.\\

     \begin{figure}[tb]
     	\centering
     	\includegraphics[height=5.5cm,width=9cm]{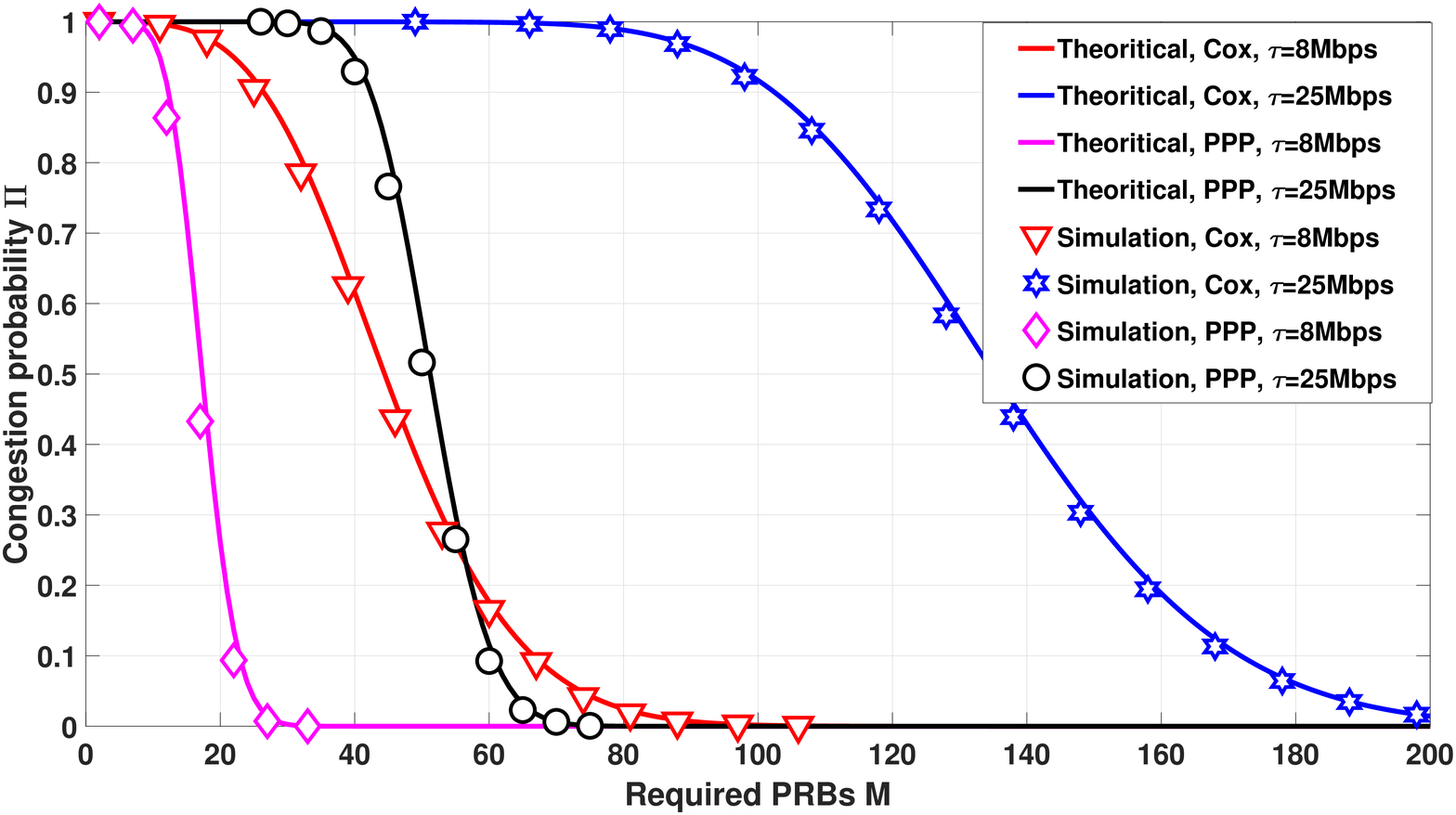}
     	\caption{Congestion probability for different values of $\tau$: comparison between Cox model and spatial PPP model}
     	\label{ths}
     \end{figure}

    During resource dimensioning process, the operator starts by defining a target congestion probability that can be tolerated for a given service. For different traffic forecasts, the number of PRBs is set to ensure that the congestion probability never exceeds its target. Fig. \ref{dimensioning1} shows the number of required PRBs that the operator should make it available when the expected cell throughput is known for two target values of congestion probability ($\Pi^*=1\%$ and $\Pi^*=5\%$) and for two road intensities ($\lambda=5km/km^2$ and $\lambda=15km/km^2$). Once again, we conclude that the number of dimensioned PRBs is very sensitive to the road intensity $\lambda$ and this is in agreement with results of Fig. \ref{ths}. Indeed, for a target congestion probability $\Pi^*=1\%$, when $\lambda$ increases from $5km/km^2$ to $15km/km^2$ (i.e., from 10 expected roads to 30) the number of dimensioned PRBs decreases by $23$, for the same cell throughput value $\tau=25 Mbps$. For a given value of $\tau$, we can notice from (\ref{cellth}) that the user intensity on roads $\delta$ is inversely proportional to roads' intensity $\lambda$. Thus for fixed $\tau$, if $\lambda$ increases $\delta$ decreases and consequently the number of required PRBs decreases. Besides, when $\lambda$ is very high, the distribution of users becomes similar to a spatial PPP. Thus, with spatial PPP model, one can have small values of dimensioned PRBs, which is considered optimistic compared to the real geometry of roads where more PRBs are required to guarantee the desired QoS.\\

     To see interference impact on dimensioning process, we divide the cell into 3 regions: cell center with a radius of R/3, cell middle represented by the ring between R/3 and 2R/3 and cell edge characterized by a distance from the BS that exceeds 2R/3. Each region of the cell experiences a given level of interference evaluated in terms of IM (Interference Margin or Noise Rise). Cell edge users always experience high interference level and IM is set to $15dB$. In cell middle we consider an interference margin of $IM=8dB$, whereas in the cell center where users perceive good radio conditions, the interference margin is set to $IM=1dB$.\\

    \begin{figure}[tb]
    	
    	\centering
    	
    	\includegraphics[height=5.5cm,width=9cm]{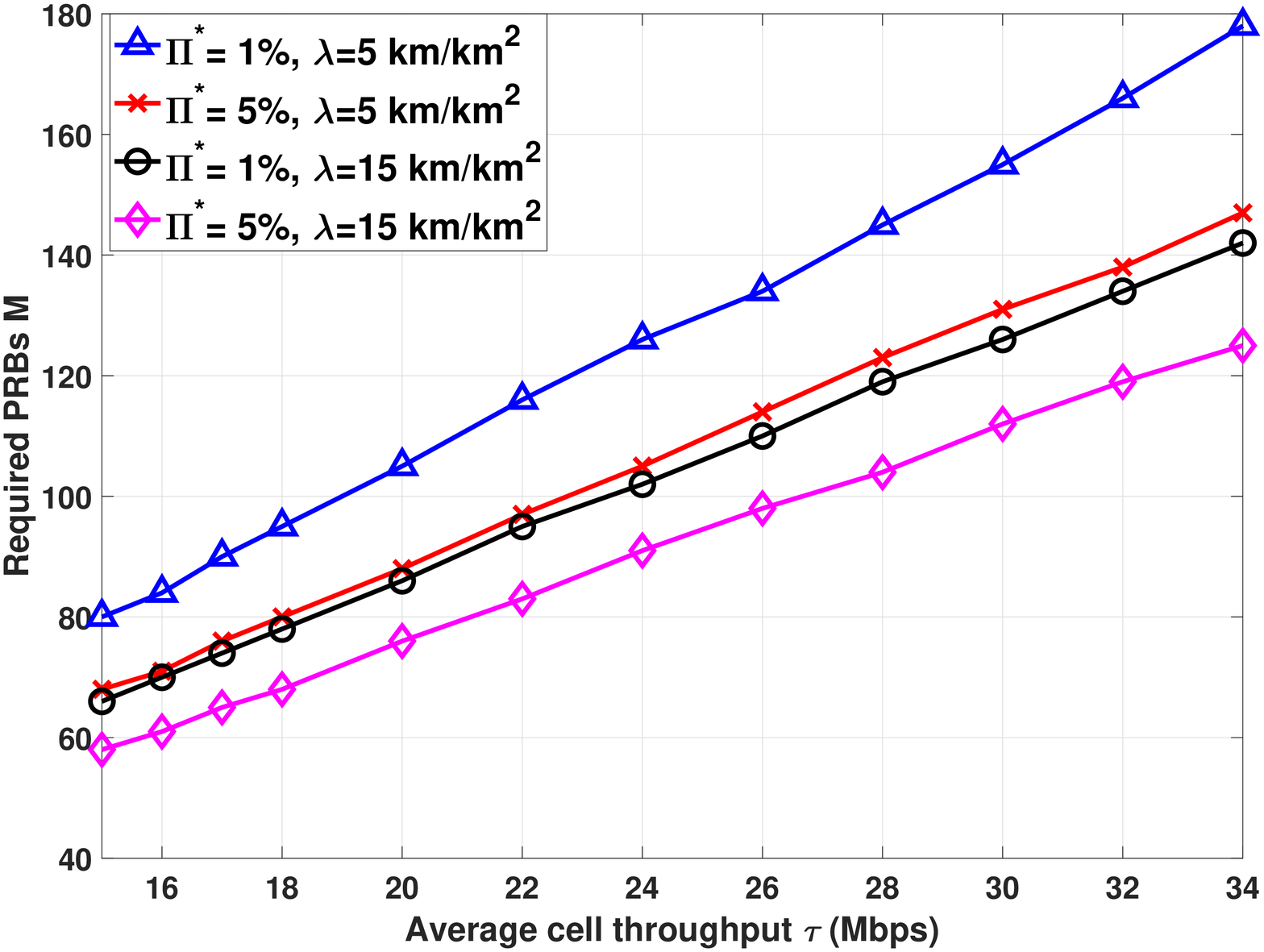}
    	\caption{Dimensioned PRBs $M$ as a function of cell Throughput $\tau$, for fixed transmission rate $C^*=1 Mbps$.}
    	\label{dimensioning1}
    \end{figure}
     
     Fig. \ref{interference1} shows the congestion probability in a noise-limited scenario (Interference level is neglected) and its comparison with the one where interference is taken in consideration as we have described above. As expected, interference has a tremendous impact on the number of required PRBs. For instance, when the target congestion probability is set to $5\%$, the number of required PRBs increases by almost 60 because of the presence of interference. Actually, interference level varies from one location to another in the same cell. Practically cell edge users experience high interference level compared to users that are close to the BS in the cell middle or cell center. Fig. \ref{interference2} shows a comparison between resources dimensioning results for the three regions of the cell: cell center, cell middle and cell edge. As we can observe, the high demand on PRBs comes especially from cell edge users that perceive bad radio conditions because of the far distance from the BS and the presence of interference. Hence, for a predicted average cell throughput, the number of dimensioned PRBs should be set considering a probable presence of traffic hotspots at the cell edge.\\

        \begin{figure}[tb]
     	\centering
     	\includegraphics[height=5.5cm,width=9cm]{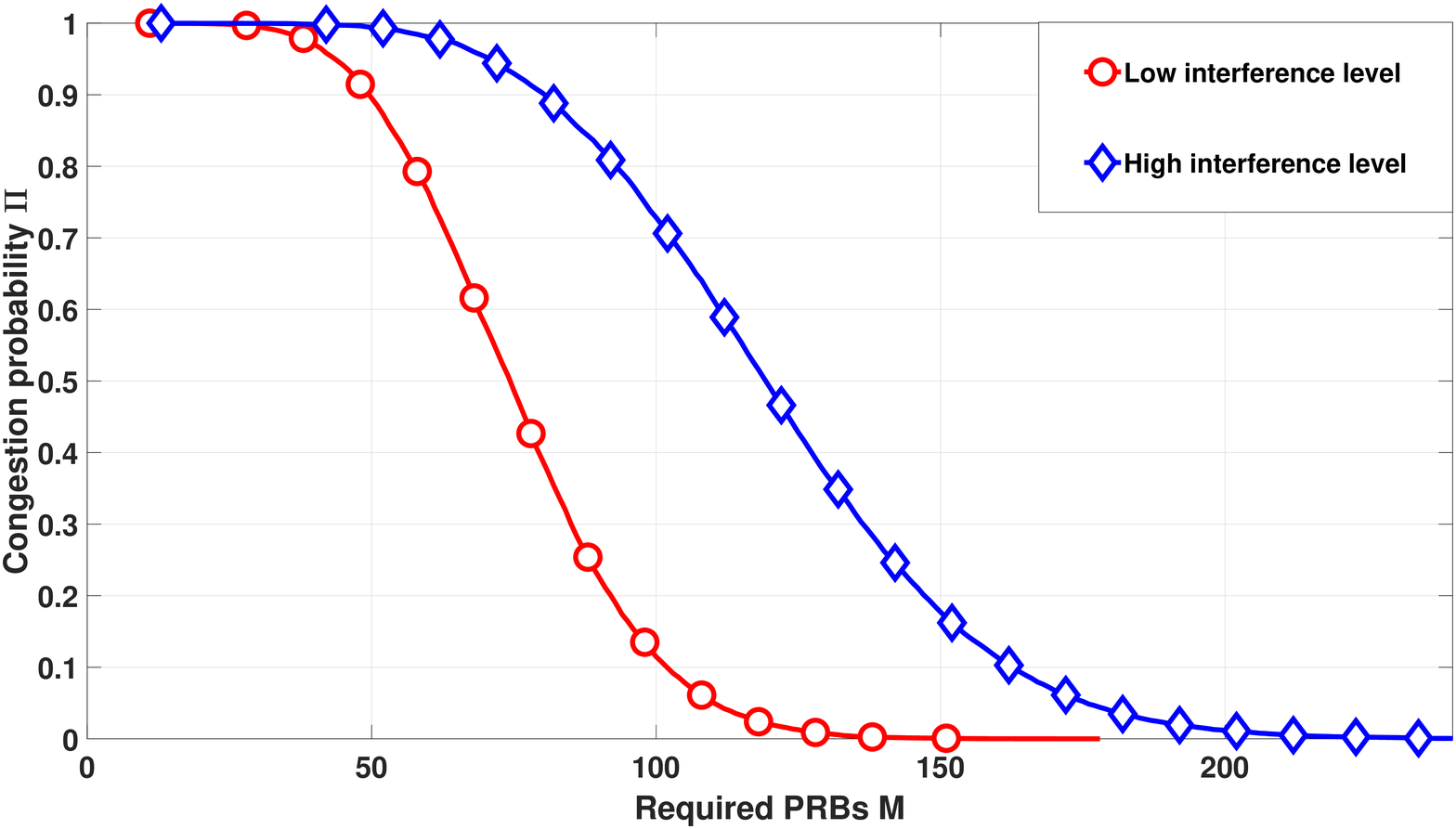}
     	\caption{Interference impact}
     	\label{interference1}
        \end{figure}
     
     \begin{figure}[tb]
     	\centering
     	\includegraphics[height=5.5cm,width=9cm]{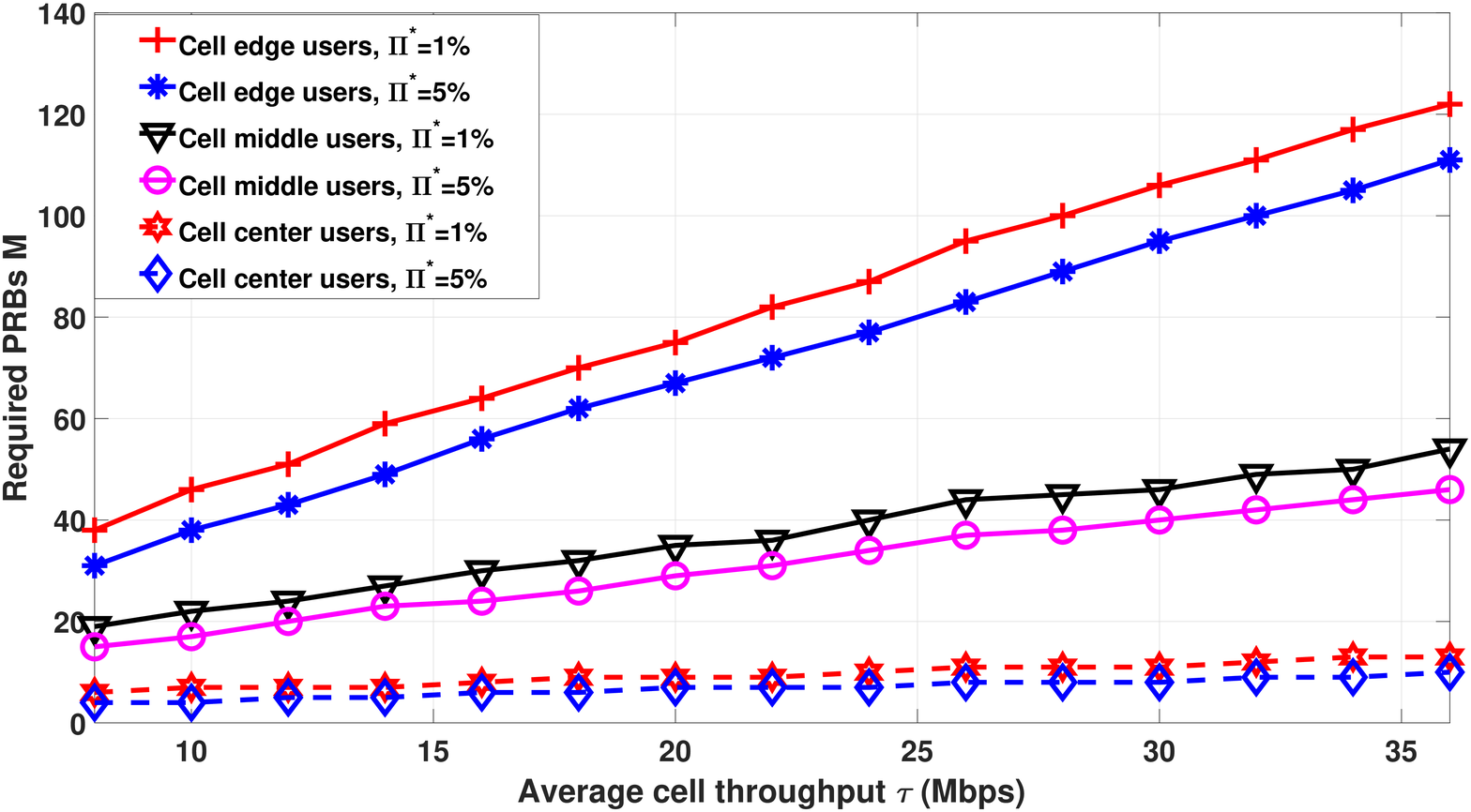}
     	\caption{Dimensioned PRBs $M$ for cell edge, cell center and cell middle users for fixed transmission rate $C^*=800 kbps$.}
     	\label{interference2}
     \end{figure}

    \section{Conclusions}
    In this paper, we have presented a resource dimensioning model for OFDM based systems that can be applied also for scalable OFDM based 5G NR interface. We have modeled the system by doubly stochastic process known as Cox point process driven by a Poisson Line Process that is used to model the random tessellation of roads covered by a typical cell. The comparison between Cox process and the spatial PPP showed that results are more optimistic with PPP. Moreover, we have derived an explicit formula of the congestion probability that can be applied to every stochastic process that describes the random distribution of a system of roads. Also, we have established an implicit relationship between the required resources and the forecast traffic for a given target congestion probability. This relationship translates the dimensioning problem that an operator can perform to look for the amount of required spectrum resources. Finally, interference impact on dimensioned resources has been evaluated especially for cell edge users. Further extension to this approach could include a combination between the spatial distribution of indoor users according to a spatial Poisson Point Process and outdoor users according to the proposed model.

    \appendices
    \section{Proof of proposition \ref{theorem1}}
    \label{proof1}
    
    To prove proposition \ref{theorem1}, first we calculate the moment generating function (i.e., Z-Transform) $f(z)$ of the discrete random variable $\Gamma$ of equation (\ref{prbtot}) conditionally on $\phi$.

   	\begin{align}
   	f(z)=\mathbb{E}(z^{\Gamma}|\phi) = \prod_{n=1}^{N}\sum_{k=0}^{+\infty} z^{nk}\mathbb{P}(X_n=k|\phi). 
    \label{p1} 
   \end{align}
   
	
  Since $X_n$ is a Poisson random variable with parameter $\mu_n$, (\ref{p1}) is simplified to \begin{equation}		
  f(z)=e^{-\delta\alpha_{N}}e^{\sum_{n=1}^{N}z^n \mu_n},
  \label{p3} 
  \end{equation}
 with $\alpha_{N}$ comes from the relation $\sum_{n=1}^{N}\mu_n= \delta \alpha_{N}$.
 
It is obvious that $f$ is analytic on $\mathbb{C}$ and in particular inside the unit circle $\omega$. Cauchy's integral formula gives then the coefficients of the expansion of $f$ in the neighborhood of $z=0$:
\begin{equation}		
\mathbb{P}(\Gamma =k|\phi)= \frac{1}{2 \pi i}\int_{\omega}\frac{f(z)}{z^{k+1}}dz.
\label{cauchy} 
\end{equation}
In (\ref{cauchy}), replacing $f$ by its expression (\ref{p3}) and parameterizing $z$ by $e^{i\theta}$ lead to
\begin{equation}		
 \mathbb{P}(\Gamma =k|\phi)= \frac{1}{2 \pi} e^{-\delta \alpha_{N}} \int_{0}^{2 \pi} \frac{e^{\sum_{n=1}^{N}\mu_n e^{i n \theta}}}{e^{i k \theta}}d\theta.
\label{p5} 
\end{equation}

Since the congestion probability is defined by the CCDF (Complementary Cumulative Distribution Function) of $\Gamma$, then

	\begin{align}
	\mathbb{P}(\Gamma\geq M|\phi) &=1-\frac{1}{2 \pi} e^{-\delta \alpha_{N}} \int_{0}^{2 \pi}e^{\sum_{n=1}^{N}\mu_n e^{i n \theta}} \sum_{k=0}^{M-1} e^{-i k \theta} d\theta.
	\label{p6} 
	\end{align}


%
The sum inside the right hand integral of (\ref{p6}) can be easily calculated to get the explicit expression of (\ref{theee}) after some simplifications.
\bibliographystyle{IEEEtran}
\bibliography{IEEEabrv,TelecomReferences}

\end{document}